\documentclass[letterpaper]{IEEEtran}
\usepackage{float}
\usepackage[english]{babel}
\usepackage{amsmath,amssymb}
\usepackage{times}
\usepackage{relsize}
\usepackage{graphicx}
\usepackage{url}
\usepackage{booktabs}
\usepackage{multirow}
\usepackage{mparhack}
\usepackage{subfigure}
\usepackage{authblk}
\usepackage{amsthm}
\usepackage[noend]{algorithmic}
\usepackage[linesnumbered,ruled,vlined]{algorithm2e}

\usepackage{array}
\usepackage{cite}
\usepackage{eqparbox}
\usepackage{mdwmath}
\usepackage{balance}
\usepackage{epsfig}
\usepackage{xcolor}
\usepackage{xfrac}
\usepackage{mathtools}
\usepackage{bm}
\usepackage{amsfonts}
\usepackage[all]{xy}
\usepackage{etoolbox}
\usepackage{graphicx}

\DeclareMathAlphabet{\mathantt}{OT1}{antt}{li}{it}
\DeclareMathAlphabet{\mathpzc}{OT1}{pzc}{m}{it}

\DeclarePairedDelimiter\abs{\lvert}{\rvert}%
\usepackage{accents}
\newcommand\munderbar[1]{%
  \underaccent{\bar}{#1}}

\newtheorem{theorem}{Theorem}

\DeclareFontFamily{OT1}{pzc}{}
\DeclareFontShape{OT1}{pzc}{m}{it}%
  {<-> s * [1.1] pzcmi7t}{}
\DeclareMathAlphabet{\mathpzc}{OT1}{pzc}%
                     {m}{it}

\DeclareMathOperator{\argmin}{\arg\min}

\makeatletter
\patchcmd{\@maketitle}
  {\addvspace{0.75\baselineskip}\egroup}
  {\addvspace{-1.45\baselineskip}\egroup}
  {}
  {}

  \def\BibTeX{{\rm B\kern-.05em{\sc i\kern-.025em b}\kern-.08em T\kern-.1667em\lower.7ex\hbox{E}\kern-.125emX}}

\makeatother

\title{Optimal Scheduling of Content Caching\\ Subject to Deadline}

\author{Ghafour Ahani and Di Yuan,
\IEEEmembership{Senior Member, IEEE}
\thanks{G. Ahani and D. Yuan are with the Department of Information Technology, Uppsala University, 751 05 Uppsala, Sweden
(e-mails: {ghafour.ahani, di.yuan}@it.uu.se).}}

\begin{document}

\maketitle
\begin{abstract}
Content caching at the edge of network is a promising technique to alleviate the burden of backhaul networks. In this paper, we consider content caching along time in a base station with limited cache capacity. As the popularity of contents may vary over time, the contents of cache need to be updated accordingly. In addition, a requested content may have a delivery deadline within which the content needs to be obtained. Motivated by these, we address optimal scheduling of content caching in a time-slotted system under delivery deadline and cache capacity constraints. The objective is to minimize a cost function that captures the load of backhaul links.
For our optimization problem, we prove its NP-hardness via a reduction from the Partition problem. For problem solving, via a mathematical reformulation, we develop a solution approach based on repeatedly applying a column generation algorithm and a problem-tailored rounding algorithm. In addition, two greedy algorithms are developed based on existing algorithms from the literature. Finally, we present extensive simulations that verify the effectiveness of our solution approach in obtaining near-to-optimal solutions in comparison to the greedy algorithms. The solutions obtained from our solution approach are within $1.6\%$ from global optimality.
\end{abstract}

\IEEEpeerreviewmaketitle
\section{Introduction}
\subsection{Motivations}
Whereas the amount of data traffic is exponentially growing, it has been realized that the major portion of the data traffic originates from duplicated downloads of a few popular contents\cite{LiQiu2017}. These duplicated downloads congest the backhaul links, hence lowering the quality of service. It is costly to increase the capacity of backhaul links, hence they should be used more effectively. A promising technique is to store the popular contents on the edge of network such as BSs with caching capability \cite{DongLiu2019,WeiJiang2017,KarthikeyanShanmugam2013}. This technique helps to improve the efficiency of communications systems via providing the contents of interest from the BSs instead of from the core network. In fact, the measurement studies in \cite{5731586,6566245} showed up to $66\%$ of traffic reduction in 3G and 4G networks via caching techniques.


Optimal content caching heavily depends on two main factors, namely the number of requests for the contents and the delivery deadlines of such requests. The number of requests for a content, referred to as the popularity of a content, may vary over time. Therefore, the contents of the cache need to be updated accordingly. An update incurs a downloading cost due to getting contents from the server to the BS cache. It is commonly assumed that a content request needs to be served as soon as it is made. We extend the problem setup and investigate a scenario in which a user can put a deadline on the delivery time of the requested content. To the best of our knowledge, the joint impact of delivery deadline and  content downloading cost in content caching has not been studied in the literature. In order to close this gap, we study content caching along time in a BS  with limited caching capacity. We address optimal scheduling of cache updates taking into account the downloading cost subject to delivery deadline and cache capacity constraints.

\subsection{Related Works}
Content caching has been studied in various system scenarios in the context of wireless communication networks. We provide a review with emphasis on the recent developments. We refer the reader to \cite{8327582} for  a comprehensive survey.


The works such as \cite{KarthikeyanShanmugam2013,7179394,7417343,7562037,7959865} studied content caching in BSs when the probability distributions of contents are known. In~\cite{KarthikeyanShanmugam2013} the objective was to minimize the expected downloading time of contents. In~\cite{7179394,7417343}~collaborative content caching among BSs was considered with the objectives of minimizing an operational cost and average downloading delay, respectively. In~\cite{7562037} decentralized content caching was studied with the presence of multi-hop communications. In~\cite{7959865} the user's hit probability was maximized.

The studies in \cite{6620380,7037523,6884109,7537180,Cost2018Deng} enhanced the system models in the works mentioned above to take into account the impact of user mobility in content caching of BSs.
The works in  \cite{6620380,7037523} took into account the movement of users where the trajectories of users are known. In \cite{6884109}, caching contents in both BSs and users was investigated with the objective of minimizing energy consumption. The works in \cite{7537180,Cost2018Deng} further improved the system model in \cite{6884109} and considered caching on mobile users such that they can obtain their contents of interest from each other via device-to-device (D2D) communications.

In contrast to the aforementioned works, the studies in \cite{6883600,7414014,7227124,7815021,7422747,8629363,8735483,8746649} investigated content caching in BSs when the popularity distributions of contents are unknown. The work in \cite{6883600} determined the popularity of a content based on the previously stored contents. The work in \cite{7414014} computed the popularity of a content using a big dataset, and proposed an optimal content caching algorithm to minimize the delivery time of contents. In \cite{7227124}, the authors estimated the popularity of contents via local interest for the content and then proposed a caching algorithm to maximize the hit rate. In \cite{7815021} an online algorithm is proposed to estimate the popularity of contents based on the incoming requests.
The works in \cite{7422747,8629363,8735483,8746649} proposed  learning-based methods to estimate the popularity of contents.

In all works mentioned so far, the popularity distribution of contents is invariant along time. The studies in \cite{7524381,8836639,8357917,8761832,8892477,ahani2019accounting,rolling2019arx} relaxed this assumption and considered content caching with time-varying popularities.
In \cite{7524381,8836639} caching contents of uniform size was studied, however, the cost of cache updating was neglected. In\cite{8357917}, the authors studied content caching in set of BSs from a learning perspective. In \cite{8761832,8892477} content caching with updates were considered in D2D and vehicle-to-vehicle networks, respectively. In~\cite{ahani2019accounting} content caching in a BS was studied in which the cost of cache updates and freshness of the contents were jointly optimized. In \cite{rolling2019arx} collaborative caching was studied, where the cost of updates is accounted for. In~\cite{ahani2019accounting,rolling2019arx}, the authors assumed a requested content needs to be served instantly after the request is made. This may not be true in some circumstances when a requester can wait before the content is delivered until a time point, that is deadline.

The works just mentioned above are the most related studies to our work in the sense that they have also considered cache updating along time. However, in these investigations either the main effort was devoted to estimating the popularity distributions of contents rather than designing effective content caching algorithms, or the cost of performing updates is neglected, or the deadlines of content requests are not considered. Therefore, we aim to complement the above works and  devote our effort to designing an effective content caching algorithm where the deadline constraints and the cost of cache updates are considered jointly.


\subsection{Our Contributions}
We investigate scheduling of content caching in a BS with limited caching capacity in a time-slotted system under delivery deadline and cache capacity constraints. Our main contribution lies on the joint consideration of time-varying popularity of contents and the deadlines of requested contents. Our objective is to optimally schedule the updates across the time slots so as to minimize the total cost of obtaining the requested contents by users. The main contributions of this work are summarized as follows:
\begin{itemize}
\item
We formally prove the NP-hardness of the problem based on a reduction from the Partition problem.
\item
We provide a mathematical problem formulation. Specifically, the problem is formulated as an integer linear program (ILP), taking into account the size of contents, capacity of the cache, deadlines of requests, and costs of content downloading and  cache updating.

\item

Based on a mathematical reformulation of the problem, we develop an effective solution approach based on a repeated column generation algorithm (RCGA). RCGA runs repeatedly and alternatively two algorithms, namely a column generation algorithm (CGA) and a problem-tailored rounding algorithm (TRA). TRA is specially designed to construct integer solutions from the fractional solutions of CGA. Moreover, RCGA provides an effective lower bound (LB) of global optimum such that the LB can be used to measure the effectiveness of any suboptimal algorithm.
\item
We propose two greedy algorithms based on existing algorithms in the literature. Even though these algorithms cannot provide high-quality solutions, they are of interest because of their low complexity and consequently fast solutions for large-scale problem instances.
\item
Finally, we conduct extensive simulations to verify the effectiveness of RCGA, and greedy algorithms by comparing them to the LB. Simulations results manifest that the solutions obtained from RCGA and the greedy algorithms are within $1\%$ and $20\%$ of global optimum, respectively.
\end{itemize}

\section{System Scenario and Complexity Analysis}

\subsection{System Scenario}\label{System_Scenario}
 The system scenario consists of a content server, a base station (BS), $U$ users within the coverage of the BS, and $F$ contents. The set of users is denoted by $\mathcal{U}=\{1,2,\dots,U\}$. The server has all the contents, and the BS is equipped with a cache of size $S$. Denote by $\mathcal{F}=\{1,2,\dots,F\}$ the set of contents.
Denote by $l_f$ the size of content $f\in \mathcal{F}$. The system scenario is shown in Fig. \ref{SystemScenario}.
\begin{figure}[ht!]
\centering
\includegraphics[scale=0.35]{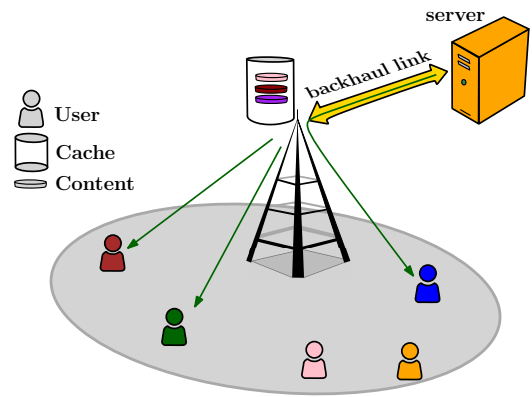}
\begin{center}
\caption{System scenario.}
\label{SystemScenario}
\end{center}
\end{figure}

We consider a time-slotted  system in which a time period is divided into $T$ time slots. Denote by $\mathcal{T}$ the set of time slots with $\mathcal{T}=\{1,2,\dots,T\}$. At the beginning of each time slot, the contents of the cache are subject to updates. Namely, some stored contents may be removed from the cache and some new contents may be added to the cache by downloading from the server.

The popularity of a content is determined by the number of requests for the content. In our model, user $u\in \mathcal{U}$, requests at most $R_u$ contents within the $T$ time slots based on its interest. The set of requests for user $u$ is denoted by $\mathcal{R}_u$. The length of a time slot is long enough to complete the downloading process of the requests from the BS or the server. We assume the time of making each request is known or can be predicted via using a prediction model \cite{Zhang2018Using}. In addition, each request has a deadline before which the requested content must be delivered to the user. For user $u$ and its $r$-th request, the requested content, the time slot of request, and the deadline of request, are denoted by $h(u,r)$, $o(u,r)$, and $d(u,r)$, respectively.

A content may become available or unavailable in the cache from a time slot to another due to cache updates. A content is either downloaded from the cache if the content is available in the cache in at least one of the time slots between $o(u,r)$ and $d(u,r)$, or, otherwise from the server. Denote by $c_s$ and $c_b$ the costs for downloading one unit of data from the server and from the cache, respectively. Intuitively, $c_s>c_b$ to encourage downloading from the cache. The time duration for downloading data from the server to the BS is neglected as the backhaul capacity is significantly higher than that of wireless access. The problem of optimally scheduling content caching subject to deadline is abbreviated to SCCD. The objective is to minimize the total cost of content downloading.

\subsection{Complexity Analysis}
In this section, we formally prove the  NP-hardness of the problem based on a reduction from the Partition problem.
\begin{theorem}
SCCD is NP-hard.
\end{theorem}
\begin{proof}
The proof is based on a polynomial-time
reduction from the Partition problem that is NP-complete \cite{garey1979computers}.
Consider a Partition problem with a set of $\mathcal{N}=\{n_1,\dots,n_N\}$ integers. The task is to determine whether it is possible to partition $\mathcal{N}$ into two subsets $\mathcal{N}_1$ and $\mathcal{N}_2$ with equal sum.

We construct a reduction from the Partition problem as follows.
We set $\mathcal{F}=\{1,\dots,N\}$, $l_f=n_f$ for $f \in \mathcal{F}$, $S=\frac{1}{2}\sum_{f \in \mathcal{F}}l_f$, and $T=1$. In this case, there is no updating cost and we only have downloading cost. The time slots of requests and deadlines for all requests are set to $1$, i.e., $o(u,r)=d(u,r)=1$ for $u \in \mathcal{U}$ and $r \in \mathcal{R}_u$.
Denote by $m_{1f}$ the number of users requesting content $f$ in this time slot. We set $m_{1f}=2$ for $f \in \mathcal{F}$, $c_s=2$, and $c_b=1$.
If content $f$ is cached, the $m_{1f}$ users can download content $f$ from the cache, thus the downloading cost for content $f$ is $m_{1f}l_fc_b+l_f(c_s-c_b)$. Otherwise, the $m_{1f}$ users have to download content $f$ from the server, giving rise to the downloading cost of $m_{1f}l_fc_s$.
That is, if the cache stores content $f$, it will obtain $m_{1f}l_fc_s-m_{1f}l_fc_b-l_f(c_s-c_b)=n_f$ gain. By this construction, the total gain that can be achieved is upper-bounded by $\frac{1}{2}\sum_{f \in \mathcal{F}}l_f$. Now the question is whether we can achieve this gain. Solving the defined instance of SCCD will answer this question and also the Partition problem. Namely, after solving this instance of SCCD, if a total gain of $\frac{1}{2}\sum_{f \in \mathcal{F}}l_f$ is achieved, then the answer to the Partition problem is yes, and the contents inside and outside the cache correspond to the two subsets $\mathcal{N}_1$ and $\mathcal{N}_2$, respectively. Otherwise, the answer to the Partition problem is no. Hence the conclusion.
\end{proof}
\section{Integer Linear Programming Formulation}
\subsection{Cost Model}
Denote by $y_{urt}$ a binary optimization variable which equals one if and only if the $r$-th request of user $u$ is downloaded in time slot $t \in \mathcal{D}_{(u,r)}=\{o(u,r),\dots,d(u,r)\}$ from the cache.
The downloading cost for user $u$ to obtain the content requested in the $r$-th request, denoted by $C_{ur}$, is expressed as:
\begin{equation}
\begin{aligned}
C_{ur}=c_bl_{h(u,r)}\sum_{t=o(u,r)}^{d(u,r)}y_{urt}+c_sl_{h(u,r)}(1-\sum_{t=o(u,r)}^{d(u,r)}y_{urt}).
\end{aligned}
\end{equation}
where the first term indicates that if the content is downloaded before its deadline from the cache, the downloading cost is $c_bl_{h(u,r)}$. Otherwise, it is downloaded from the server with cost $c_sl_{h(u,r)}$.
The downloading cost for completing all requests of user $u$, denoted by $C_u$, is:
\begin{equation}
\label{Ci}
C_u=\sum_{r=1}^{R_u}C_{ur}.
\end{equation}

Thus, the downloading cost for completing all requests for all users, denoted by $C_{download}$, is expressed as:
\begin{equation}
\begin{aligned}
C_{download}=\sum_{u=1}^{U}C_u.
\end{aligned}
\end{equation}
For the cache, the cost due to cache updates is referred to as the updating cost. This cost over the time slots, denoted by $C_{update}$, is expressed as:
\begin{equation}
\begin{aligned}
C_{update}=\sum_{t=1}^{T}\sum_{f=1}^{F}l_f(c_s-c_b)a_{tf},
\label{c_update}
\end{aligned}
\end{equation}
where $a_{tf}$ is a binary variable which equals one if and only if the cache does not store content $f$ in time slot $t-1$, but stores the content in time slot $t$, and $l_f(c_s-c_b)$ is the cost for downloading content $f$ from the server to the cache.
\subsection{Problem Formulation}

In general, as the popularity of contents changes over time,
storing popular contents in each time slot will reduce the downloading cost, but it significantly increases the updating cost. On the other hand, if the stored contents remain unchanged over the time slots, the updating cost is low, but the downloading cost will be high. Based on this, our optimization problem is to minimize the total cost consisting of the downloading and the updating cost by optimizing decisions in terms of caching the contents over the time slots.
Denote by $\bm{x}$ an $F\times T$ matrix of optimization variables for $F$ contents and $T$ time slots:
\[
\bm{x}=\{x_{tf}, t\in \mathcal{T}~\text{and}~f\in \mathcal{F} \}.
\]
where $x_{tf}$ is a binary variable that takes value one if and only if content $f$ is stored in time slot $t$. SCCD can be formulated as an integer linear program (ILP) and shown in (\ref{ILP}).

\begin{figure}[!h]
\begin{subequations}
\begin{alignat}{2}
\text{(ILP)}~~ &\min\limits_{\bm{x},\bm{a},\bm{y}}\quad  C_{download}+C_{update}\\
\text{s.t}. \quad
& \sum_{f\in \mathcal{F}}x_{tf}l_f \leq S,t\in \mathcal{T} \label{const:cacheSize}\\
& a_{tf} \ge x_{tf}-x_{(t-l)f} ,t\in \mathcal{T}\setminus{\{1\}},f\in \mathcal{F}\label{const:a1}\\
& a_{tf} \le 1-x_{(t-1)f},t\in \mathcal{T}\setminus{\{1\}},f\in \mathcal{F}\label{const:a2}\\
& a_{tf} \le x_{tf},t\in \mathcal{T}\setminus{\{1\}},~f\in \mathcal{F}\label{const:a3}\\
& a_{1f}=x_{1f},f\in \mathcal{F}\label{const:a4}\\
&y_{urt}\le x_{th(u,r)}, u \in \mathcal{U}, r \in \mathcal{R}_u,t \in \mathcal{D}_{(u,r)}\label{const:y1}\\
&\sum_{t=o_{ur}}^{d_{ur}}y_{urt}\le1, u \in \mathcal{U}, r \in \mathcal{R}_u\label{const:y2}\\
& x_{tf},a_{tf}\in \{0,1\},t\in \mathcal{T},f\in \mathcal{F}\\
& y_{urt}\in \{0,1\},u\in \mathcal{U},r \in \mathcal{R}_u,t\in \mathcal{D}_{(u,r)}.
\end{alignat}
\label{ILP}
\end{subequations}
\end{figure}
Constraints (\ref{const:cacheSize}) indicate that the total amount of cache space used for storing the contents is less than or equal to the cache capacity in each time slot.
Constraints \eqref{const:a1}, \eqref{const:a2}, \eqref{const:a3}, and \eqref{const:a4} together ensure that $a_{tf}$ is one if and only if the cache does not store content $f$ in time slot $t-1$, but stores the content in time slot $t$.
Constraints (\ref{const:y1}) state that $y_{urt}$ can take value one only if $x_{th(u,r)}=1$, i.e., content $h(u,r)$ is stored in the cache in time slot $t$. Constraints (\ref{const:y2}) say that request $r$ from user $u$ is met in at most one of the time slots between the time slot of request and its deadline.

ILP~\eqref{ILP} can be solved by an off-the-shelf integer programming algorithm from optimization packages. However, for large-scale problem instances solving the problem needs significant computational effort. Therefore, we develop a column generation algorithm and rounding mechanism, presented in Section \ref{alg_design}, to obtain near-to-optimal solutions of SCCD.

\section{Problem Reformulation}
In this section, we provide a reformulation of SCCD that enables a solution approach based on column generation. We will see in Section~\ref{sec:performance} that the algorithm achieves near-to-optimal solutions.

We define sequence $\bm{x}_f=[x_{1f}, x_{2f},\dots,x_{Tf}]^ \mathrm{ T } $ to represent the caching solution of content $f$ over the $T$ time slots.
As $x_{tf}\in\{0,1\}$ for $t \in \mathcal{T}$, in total $K=2^T$ possible sequences exist for content $f$. However, as will be clear later on, the algorithm needs to deal with only a small subset of the candidate sequences.
Denote by $\mathcal{K}$ a set, with $\mathcal{K}=\{1,2,\dots,K\}$. Denote by $w_{fk}$ a binary variable where $w_{fk}=1$ if and only if the $k$-th sequence of content $f$ is selected, otherwise zero. Exactly one of them is used in the solution of the problem, thus $\sum_{k=1}^{K}w_{fk}=1$.
For any given sequence, the total cost of the sequence can be calculated as the sequence contains known caching decisions. The total cost for content $f$ with respect to the $k$-th sequence is denoted by $C_{fk}$ and is expressed in (\ref{Cfk}). Denote by constants $x_{tf}^{(k)}$, $y^{(k)}_{urt}$, and $a^{(k)}_{tf}$ the values of $x_{tf}$, $y_{urt}$, and $a_{tf}$ with respect to the $k$-th sequence, respectively. Note that given the values of $x_{tf}^{(k)}$ the value of $y_{urt}^{(k)}$ can be determined.
\begin{equation}
\begin{aligned}
C_{fk}&=\sum_{u=1}^{U}\sum_{r=1}^{R_u} l_{h(u,r)}[c_b\sum_{t=o(u,r)}^{d(u,r)}y^{(k)}_{urt}+c_s(1-\sum_{t=o(r,h)}^{d(r,h)}y^{(k)}_{urt})]\\
&+\sum_{t=1}^{T}l_fc_sa^{(k)}_{tf}.\label{Cfk}
\end{aligned}
\end{equation}

Based on the above notion, SCCD is reformulated as \eqref{PR}. Constraints~\eqref{MP_C1} formulate cache capacity over the time slots. These constraints have the same meaning as constraints \eqref{const:cacheSize}. Constraints~\eqref{MP_C1} say that exactly one sequence has to be selected for each content. In formulation~\eqref{PR} the deadline and updating constraints (i.e., constraints~\eqref{const:a1}-\eqref{const:y2}) are not present, and they are embedded in the sequences. As can be seen both \eqref{ILP} and \eqref{PR} are valid optimization formulations of SCCD. However they differ in structure.
\begin{figure}[!h]
\vskip -1pt
\begin{subequations}
\begin{alignat}{2}
~~~~~~~~~~&\min\limits_{\bm{w}}\quad  \sum_{f\in \mathcal{F}}\sum_{k\in \mathcal{K}}C_{fk}w_{fk} \label{MPC1}\\
\text{s.t}. \quad
&  \sum_{f\in \mathcal{F}}\sum_{k\in \mathcal{K}}l_fx^{(k)}_{tf}w_{fk} \leq S,t\in \mathcal{T} \label{MP_C1}\\
&\sum_{k\in \mathcal{K}}w_{fk}=1,f\in \mathcal{F}\\
& w_{fk}\in \{0,1 \},f\in \mathcal{F},k\in \mathcal{K}. \label{MP_C2}
\end{alignat}
\label{PR}
\vskip -20pt
\end{subequations}
\end{figure}

\section{Algorithm Design}\label{alg_design}
In this section, we present our solution approach.
We first consider the continuous version of formulation~\eqref{PR} and apply column generation to derive its global optimum. This gives obviously a lower bound to the global optimum of SCCD. Next, if the solution obtained from the column generation algorithm (CGA) is fractional, we use a tailored rounding algorithm (TRA) to obtain integer solutions. Using TRA, some of the decisions in terms of caching will be fixed and CGA will be used again to resolve the new problem subject to these decisions. This process will continue until an integral solution is obtained. We refer to this solution approach as repeated column generation algorithm (RCGA).

\subsection{Column Generation Algorithm}
For some structured linear programming problems, column generation can reduce the computational complexity for solving large-scale scenarios~\cite{viableCG}. The main advantage of using column generation is that the optimal solution can be obtained without the need of considering the set of all possible columns of which the number is typically exponentially many. In column generation, the problem under consideration is decomposed into a so called master problem (MP) and a subproblem (SP). The algorithm iterates between a restricted MP (RMP) and SP. The idea is to start with a very limited set of columns. The algorithm solves the SP to generate one or multiple new columns that improve the objective function of the RMP. This process is repeated until no improving column exists. In SCCD, a column is defined as a value assignment of sequence $[x_{1f}, x_{2f},\dots,x_{Tf}]^ \mathrm{ T }$.
\subsubsection{MP and RMP}
MP is the continuous version of formulation \eqref{PR}. CGA starts with a small subset $\mathcal{K}^\prime_f\subset\mathcal{K}$ for any content $f\in \mathcal{F}$. This leads to a so-called restricted version of the MP problem referred to as RMP, which is expressed in \eqref{RMP}. Denote by $K^\prime_f$ the cardinality of $\mathcal{K}^\prime_f$.
\begin{figure}[!h]
\vskip -2pt
\begin{subequations}
\begin{alignat}{2}
\text{(RMP)}~~~~~~~~~~&
\min\limits_{\bm{w}}\quad  \sum_{f\in \mathcal{F}}\sum_{k\in \mathcal{K}^\prime_f}C_{fk}w_{fk}
 \label{obj:RMP} \\
\text{s.t}. \quad
&  \sum_{f\in \mathcal{F}}\sum_{k\in \mathcal{K}^\prime_f}l_fx^{(k)}_{tf}w_{fk} \leq S,t\in \mathcal{T}\label{const:cachesize}\\
&\sum_{k\in \mathcal{K}^\prime_f}w_{fk} = 1,f\in \mathcal{F}\label{const:1col}\\
& 0\le w_{fk} \le 1,f\in \mathcal{F},k\in \mathcal{K}^\prime_f.
\end{alignat}
\label{RMP}
\vskip -2pt
\end{subequations}
\end{figure}

\subsubsection{Subproblem}
The SP uses the dual optimal solution to generate new columns.
Denote by $\bm{w}^*$ the optimal solution of $\eqref{RMP}$. Denote by $\bm{\pi}^*$ and $\bm{\beta}^*$ the optimal values of the corresponding dual variables of constraints \eqref{const:cachesize} and \eqref{const:1col}, respectively.
Here, $\bm{w}^*=\{w^*_{fk}, f\in \mathcal{F} ~\text{and}~ k\in \mathcal{K}^\prime_f\}$,
$\bm{\pi}^*=[\pi^*_1,\pi^*_2,\dots,\pi^*_{T}]^\mathrm{ T }$ and $\bm{\beta}^*=[\beta^*_1,\beta^*_2,\dots,\beta^*_F]^\mathrm{ T }$.
After obtaining $\bm{w}^*$, checking if $\bm{w}^*$ is the optimum of MP can be determined by finding a column with the minimum reduced cost for each content $f\in \mathcal{F}$. If all these values are nonnegative, then the current solution is optimal. Otherwise, we add the columns with negative reduced costs to their respective sets.

Given $(\bm{\pi}^*,\bm{\beta}^*)$,  the reduced cost of content $f\in \mathcal{F}$ for column $\bm{x}_f=[x_{1f},x_{2f},\dots,x_{Tf}]$ is $C_{f}-\sum_{t=1}^{T}l_f\pi^*_t x_{tf}-\beta^*_f$. Here, $C_{f}$ is expression (\ref{Cfk}) in which $y_{urt}^{(k)}$ and $a_{tf}^{(k)}$ are replaced with their counterparts of optimization variables.
To find the column with minimum reduced cost for content $f\in \mathcal{F}$, we need to solve subproblem SP$_f$, shown in \eqref{SP}. Denote by $\bm{x}^*_f$ the optimal solution of SP$_f$, i.e., $\bm{x}^*_f=[x^*_{1f},x^*_{2f},\dots,x^*_{Tf}]^\mathrm{T}$.
If the reduced cost of $\bm{x}^*_f$ is negative, we add $\bm{x}^*_f$ to $\mathcal{K}^\prime_f$. Note that term $-\beta_f^*$ is a constant and thus dropped from the objective function.

\begin{figure}[!h]
\vskip -20pt
\begin{subequations}
\begin{alignat}{2}
(\text{SP}_f)~~~~~~~&
\min\limits_{\bm{x},\bm{a},\bm{y}}\quad  C_{f}-\sum_{t=1}^{T}l_f\pi^*_t x_{tf} \label{SP_objective}\\
\text{s.t}. \quad
& a_{tf} \ge x_{tf}-x_{(t-1)f},t\in \mathcal{T}\setminus{\{1\}}\\
& a_{tf} \le x_{tf},t\in \mathcal{T}\setminus{\{1\}}\\
& a_{tf} \le 1-x_{(t-1)f},t\in \mathcal{T}\setminus{\{1\}}\\
& a_{1f}=x_{1f}\\
&y_{urt}\le x_{th(u,r)}, u \in \mathcal{U}, r \in \mathcal{R}_u,t \in\mathcal{D}_{(u,r)}\label{constSP:y1}\\
&\sum_{t=o_{ur}}^{d_{ur}}y_{urt}\le1, u \in \mathcal{U}, r \in \mathcal{R}_u\label{constSP:y2}\\
& x_{tf}, a_{tf}\in \{0,1 \},t\in \mathcal{T}\\
& y_{urt}\in \{0,1\},u\in \mathcal{U},r \in \mathcal{R}_u,t\in\mathcal{D}_{(u,r)}.
\end{alignat}
\label{SP}
\vskip -20pt
\end{subequations}
\end{figure}
\begin{figure*}[!t]
\normalsize
\begin{equation}
\label{eq:10}
\begin{aligned}
&C_{f}-\sum_{t \in \mathcal{T}}l_f\pi^*_{tf} x_{tf}\\
&=\sum_{u \in \mathcal{U}}\sum_{r \in \mathcal{R}_u: h(u,r)=f}l_f[ \sum_{k=o(u,r)}^{d(u,r)}y_{urk}c_b
+(1- \sum_{k=o(u,r)}^{d(u,r)}y_{urk})c_s]+\sum_{t=1}^{T}l_f(c_s-c_b)a_{tf}-\sum_{t=1}^{T}l_f\pi^*_{tf} x_{tf}\\
&=\underbrace{\sum_{u\in \mathcal{U}}\sum_{r\in \mathcal{R}_u:h(u,r)=f}l_fc_s}_{Q_f}+
\left[\underbrace{l_f(c_s-c_b)}_{q_{f}}a_{2f}
+\underbrace{(-l_f\pi^*_2)}_{p_{2f}} x_{2f}\right]-
\underbrace{\left[\sum_{u \in \mathcal{U}}\sum_{\substack{r \in \mathcal{R}_u:\\h(u,r)=f\\o(u,r)=1\text{~or~}2 \\d(u,r)\ge2}}l_f(c_s-c_b)y_{urk} \right]}_{\sum_{i=1}^2g_{it}^{\ge2}}.
\end{aligned}
\end{equation}
\hrulefill
\vspace*{4pt}
\end{figure*}
\begin{figure*}
\centering
\includegraphics[scale=0.5]{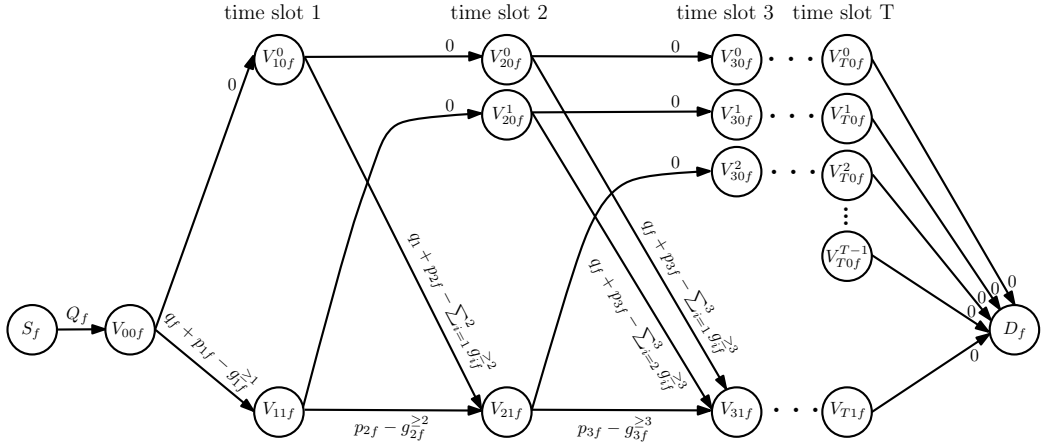}
\begin{center}
\vspace{-4mm}
\caption{Graph of the shortest path problem for SP\text{$_f$}.}
\label{SP1}
\end{center}
\vspace{-8mm}
\hrulefill
\end{figure*}
 Even though SP$_f$ is an ILP, we show that it can be solved in polynomial time by mapping to a shortest path problem.
\subsection{Subproblem as a Shortest Path Problem}
For SP\text{$_f$}, we construct an acyclic directed graph where finding the shortest path from defined
 source to distention is equivalent to solving the subproblem.
Denote by $Q_f$ the total downloading cost for content $f$ when all requests over all time slots are served from the server, i.e., $Q_f=\sum_{u \in \mathcal{U}}\sum_{r \in \mathcal{R}_u:h(u,r)=f}l_fc_s$. Denote by $q_{f}=l_f(c_s-c_b)$
 the updating cost when the content is not stored in the
previous time slot, but is stored in the current time slot. Denote by
$p_{tf}=-l_f\pi^*_t$ the cost related to the dual optimal solution in time slot $t$. Denote by $g_{tf}^{\ge d}$ the cost of the requests made for content $f$ in time slot $t$ with deadline greater than or equal to time slot $d$, that is:
\begin{equation}
\begin{aligned}
g_{tf}^{\ge d}=\sum_{u \in \mathcal{U}}\sum_{\substack{r \in \mathcal{R}_u:\\
h(u,r)=f\\
o(u,r)=t\\
d(u,r)\ge d}}l_f(c_s-c_b).
\end{aligned}
\end{equation}

The graph is shown in Fig.~\ref{SP1}. We first introduce the vertices and then the arcs. Two vertices $S_f$ and $D_f$ are defined to represent the source and destination, respectively.
$V_{00f}$ is a vertex representing $x_{0f}=0$. For time slot $t\in \mathcal{T}$, in total $t+1$ vertices are defined, represented by $V_{t1f}$ and  $V_{t0f}^{k}$, $k \in \{0,\dots,t-1\}$. Vertex $V_{t1f}$ represents decision $x_{tf}=1$ and vertices $V_{t0f}^{k}$, $k \in \{0,\dots,t-1\}$, represent decision $x_{tf}=0$ for the following scenarios. Vertex $V_{t0f}^0$ indicates that the content has not been stored in the cache in time slots $1,\dots,t$, i.e., $x_{jf}=0$ for $j \in \{1,\dots,t\}$. Vertex $V_{t0f}^k$, $k \in \{1,\dots,t-1\}$, indicates the content has been in the cache in time slot $k$, but not in the subsequent time slots until time slot $t$, i.e., $x_{kf}=1$ and $x_{jf}=0$ for
$j \in \{k+1,\dots,t\}$. These vertices are defined to trace the most recent time slot that the content was in the cache. Tracing enables to define the cost of each arc with respect to deadline.

Now, we introduce the arcs and their weights. There is an arc from $S_f$ to $V_{00f}$ with weight $Q_f$. For time slot $1$, there are two outgoing arcs from $V_{00f}$, one to $V_{11f}$ with weight $q_{f}-p_{1f}-g_{1f}^{\ge1}$ and the other to $V_{10f}^0$ with weight zero. Consider time slot $t\in\{2,\dots,T\}$, for vertex $V_{t1f}$ there are $t$ incoming arcs such that one comes from $V_{(t-1)1f}$ with weight $p_{tf}-g_{tf}^{\ge t}$, and the others come from $V_{(t-1)0f}^k$ for $k \in \{0,\dots,t-2\}$ with weight $q_{f}+p_{tf}-\sum_{i=k+1}^{t}g_{if}^{\ge t}$, respectively. Selecting vertex $V_{(t-1)0f}^k$ in the path means that no request has been served in time slots $k+1,\dots,t$ as $x_{jf}=0$ for $j \in \{k+1,\dots,t\}$, hence the third term in the weight is defined to serve all requests that are made in time slots $k+1,\dots,t$ with deadline later than or equal to time slot $t$. For each vertex $V_{t0f}^{i}$, $i \in \{0,\dots,t-2\}$, there is one incoming arc from $V_{(t-1)0f}^{i}$ with weight zero. For vertex $V_{t0f}^{t-1}$ the arc comes from $V_{(t-1)1f}$ with weight zero. There are $T+1$ arcs from vertices $V_{T1f}$  and $V_{T0f}^i$ to $D_f$ all having weight zero.

\begin{theorem}
For each content $f \in \mathcal{F}$, SP$_f$ can be solved in polynomial time as a shortest path problem.
\label{shortest_path}
\end{theorem}
\begin{proof}
We show that the optimal solution of the subproblem can be obtained from the shortest path of the graph defined above.
Assume the optimal solution of SP\text{$_f$}, i.e., $\bm{x}^*$, $\bm{a}^*$, and $\bm{y}^*$ are given. The path is constructed as follows. One of the following three scenarios may happen in time slot $t\in \mathcal{T}$. First, if $x_{tf}=1$, the vertex on the path is $V_{t1f}$. Second, if $x_{1f}=\dots=x_{tf}=0$, the next vertex is $V_{t0f}^0$. Third, if $x_{if}=1$ for time slot  $i\in\{1,\dots,t-1\}$ and $x_{jf}=0$ for all $j=i+1,\dots,t$, the next vertex is $V_{t0f}^i$.  By construction of the graph, this path from $S_f$ to $D_f$ gives the same objective function of SP\text{$_f$} as $\bm{x}^*$, $\bm{a}^*$, and $\bm{y}^*$.

Conversely, assume the shortest path is given. For time slot $t$, if the path contains one of the vertices $V_{t0f}^i$ for $i \in \{0,\dots,t-1\}$, we set $x_{tf}=0$. Otherwise, the path contains vertex $V_{t1f}$, and we set $x_{tf}=1$. As soon as the values of $x_{tf}$ for $t \in \mathcal{T}$ and $f \in \mathcal{F}$ are known, the values of $a_{tf}$ for $t \in\mathcal{T}$ and $f\in \mathcal{F}$ and $y_{urt}$ for $u \in \mathcal{U}$, $f \in \mathcal{F}$, and $t \in \mathcal{D}_{(u,r)}=\{o(u,r),\dots,d(u,r)\}$ can be easily determined. The value of $y_{urt}$ is set to the first time slot that the request can be served. By the construction of the graph, this solution gives the same objective function value as the shortest path. To clarify why this is correct we give an example. Assume that the shortest path $S_f,V_{00f},V_{10f}^{0},V_{21f},V_{30f}^{2},...,D_f$ is given which has length $Q_f+q_f+a_{2f}-\sum_{i=1}^{2}g_{if}^{\ge2}$. Then, we set $x_{tf}=0$ for $t \in \mathcal{T}\setminus \{2\}$ and $x_{2f}=1$, $a_{2f}=1$, and $y_{urt}=1$ for all requests that can be served in time slot~2. With these setting of variables, the objective function has the same value as the length of the shortest path, as shown in \eqref{eq:10}. Based on the rationale illustrated in the example, it is straightforward to conclude the correctness in general.

 Finally, the shortest path problem can be solved in polynomial time~\cite{Cormen2009introduction}.
Hence, the conclusion.
\end{proof}

\begin{algorithm}\label{alg_CGA}
\caption{Column Generation Algorithm (CGA)}
\begin{algorithmic}[1]
\algsetup{linenosize=\tiny}
\small
\REQUIRE $S$, $c_b$, $c_s$, $l_f$ for $f \in \mathcal{F}$, $o(u,r)$, $h(u,r)$ and $d(u,r)$ for $t \in \mathcal{T}, u \in \mathcal{U}, f \in \mathcal{F}, r\in \{1,\dots,R_u\}$
\ENSURE $\bm{w}^*$
\STATE $\mathcal{K}^\prime_f \leftarrow \{\bf{0}^\mathrm{T}\}$, $f \in \mathcal{F}$
\STATE STOP $\leftarrow 0$
\WHILE{(STOP$=0$)}
\STATE Solve RMP and obtain $\bm{w}^*$ and $(\bm{\pi}^*,\bm{\beta}^*)$
\STATE STOP $\leftarrow 1$
\FOR{$f=1$ ~to~ $F$}
\STATE Solve SP$_f$ using $(\bm{\pi}^*,\bm{\beta}^*)$ and obtain $\bm{x}^*_f$
\IF {$C_{f}^*-\sum_{t=1}^{T}l_f\pi^*_t x^*_{tf}-\beta^*_f <0$}
\STATE $\mathcal{K}^\prime_{f} \leftarrow \mathcal{K}^\prime_{f}\cup \{\bm{x}^*_f\}$
\STATE STOP $\leftarrow 0$
\ENDIF
\ENDFOR
\ENDWHILE
\STATE Return $\bm{w}^*$ as the optimal solution
\end{algorithmic}
\end{algorithm}

\subsection{Rounding Algorithm}
As the solution obtained from the RMP (i.e., $\bm{w}^*$)  may be fractional, we need a mechanism to obtain a feasible integer solution. One straightforward way is to round the fractional elements of $\bm{w}^*$. However, this way of rounding has some limitations. First, the solution may easily become infeasible. Second, even if the solution is feasible, it may be far from the global optimum. Third, when an element of $\bm{w}^*$, say $w_{fk}$, becomes fixed in value, the caching decisions of content $f$ for all time slots are made, and consequently there is no opportunity to improve the solution of content $f$.

In order to overcome the above limitations, we make a rounding decision for one content and one time slot
at a time. More specifically, the caching decision of content $f$ in time slot $t$ is made based on the value of $z_{tf}$, and  $z_{tf}$ is the sum of those elements of $\bm{w}^*$ such that the corresponding columns store content $f$ in time slot $t$, that is, $z_{tf}=\sum_{k\in \mathcal{K}^\prime_f}x^{(k)}_{tf}w^*_{fk}$. In fact, the value of $z_{tf}$ can be viewed as an indicator of how probable it is to store content $f$ in time slot $t$ at optimum. In the following we prove a relationship between $\bm{z}$ and $\bm{w}^*$ and then base our algorithm on this result.

\begin{theorem}
For any content $f\in \mathcal{F}$ and $k \in\mathcal{K}_f$, $w^*_{fk}$ is binary if and only if every element of $\bm{z}_{f}$ is binary, where $z_f=[z_{1f},z_{2f},\dots,z_{Tf}]$.
\label{IntegerTheory}
\end{theorem}

\begin{proof}
For necessity, for any content $f\in \mathcal{F}$, if $w^*_{fk}$ is binary for any $k$, $k\in \mathcal{K}^\prime_f$, it is obvious that all elements of $\bm{z}_{f}$ are binary.
Now, we prove the sufficiency.
For any content $f\in \mathcal{F}$, assume that every element in $\bm{z}_{f}$ is binary. Assume that $w^*_{fk}$ is larger than zero for $k\in \mathcal{K}^{\prime\prime}_f\subseteq \mathcal{K}^\prime$.
As element $z_{tf}=\sum_{k\in \mathcal{K}^{\prime\prime}_f}x^{(k)}_{tf}w^*_{fk}$ is either zero or one, the value of $x_{tf}^{(k)}$ for $k\in \mathcal{K}^{\prime\prime}_f$ must be either all zero or all one. Otherwise, as $\sum_{k \in \mathcal{K}^{\prime\prime}_f}w_{fk}=1$, one of the elements of $\bm{z}_f$ will become fractional. This means that all columns corresponding to $w_{fk}^*$
for $k\in \mathcal{K}^{\prime\prime}_f$ must be the same. Having two columns with the same values violates the fact that the sequences of any two $w^*_{fk}$ differ in at least one element. Therefore, for any content $f\in \mathcal{F}$, if $z_{tf}$ is binary for any $t\in \mathcal{T}$, then $w^*_{fk}$ is an binary for any $k\in \mathcal{K}^\prime_f$. Hence the proof.
\end{proof}

A family of  rounding algorithms can be derived based on how the caching decisions of the contents are made. We do it gradually. First, for content $f$ and time slot $t$, if $z_{tf}=1$ then the decision is to store this content in this time slot, i.e., $x_{tf}=1$. Next, we find the fractional element of $\bm{z}$ being closest to zero or one, and round the value, giving  the caching decision of the corresponding content and time slot. Next, the CGA will be applied subject to the rounded values to obtain the new $\bm{w}^*$. This process is repeated until a feasible integer solution is obtained. Note that a caching decision for a content and time slot, once made, will remain in all the subsequent iterations. An important observation is that the SP$_f$, $f \in \mathcal{F}$, with the giving caching decisions still can be solved via shortest path. If $x_{tf}=1$, we simply remove vertices $V_{j0}^i$, for $j=t,\dots,T$ and $i=1,\dots,t$, and the arcs connected to these vertices from the graph. If $x_{tf}=0$, we remove vertex $V_{t1}$ and its connected arcs.

TRA is presented in Algorithm~\ref{alg_ERA}.
Symbol~$\leftarrow$ is used when a value is assigned to a programming variable and symbol~$\leftleftarrows$ is used when an optimization variable is fixed to a value.
The details of TRA are as follows.
First, in Line~$1$, $\bm{z}$ is calculated.
For each $t \in \mathcal{T}$ and $f \in \mathcal{F}$, if $z_{tf}$ has value one, then TRA fixes $x_{tf}=1$ in SP$_f$ by Line~\ref{fixxto1}. In addition, as $x_{tf}$ is fixed to one, the columns in $\mathcal{K}^\prime_f$ that have value zero in time slot $t$ cannot be used any more and they are discarded. To achieve this, we fix $w_{fk}=0$, $k \in \mathcal{K}^\prime_f$, if $ x^{(k)}_{tf}=0$. This is done by Line $\ref{fixyto0}$.

Second, as long as $\bm{w}^*$ is not an integer solution, then by Theorem~\ref{IntegerTheory} at least one element of $\bm{z}$ must be fractional. The fractional value of $\bm{z}$ being nearest to zero, its corresponding time slot, and content are calculated by Lines~\ref{minz}-\ref{minzloc}, and these are denoted by $\munderbar{z}$, $\munderbar{t}$, and $\munderbar{f}$ respectively. Likewise,
 the fractional value of $\bm{z}$ being nearest to one, its corresponding time slot, and content are calculated by Lines~\ref{maxz}-\ref{maxzloc}, and these are denoted by $\bar{z}$, $\bar{t}$, and $\bar{f}$ respectively. If $\munderbar{z}$ is less than $\bar{z}$, TRA fixes the value of time slot $\munderbar{t}$ to zero by Line~\ref{fixxunderbarto0}. Furthermore, those columns not compatible with the decision are discarded from $\mathcal{K}^\prime_{\bar{f}}$. This is done by Line~\ref{fixyunderbarto0}. Otherwise, TRA checks whether there is enough spare space to store content $\bar{f}$. If yes, then the value of time slot $\bar{t}$ is fixed to one in SP$_{\bar{f}}$ by Line~\ref{fixxbarto1}, and the columns with value zero in time slot $\bar{t}$  are discarded from $\mathcal{K}^\prime_{\bar{f}}$ by Line~\ref{fixybarto0}. If no, the value of time slot $\bar{t}$ is fixed to zero by Line~\ref{fixxbarto0} and the columns with value one in time slot $\bar{t}$ are discarded from $\mathcal{K}_{\bar{f}}$ by Line~\ref{fixybarto01}.

Third, TRA fixes $x_{tf}=0$ for the contents that have size larger than the remained spare cache space. This is done by Lines~\ref{fixto0bysize1}-\ref{fixto0bysize2}.

Finally, the above operations may lead to discarding all columns of a content such that the RMP becomes infeasible. To avoid this, an auxiliary column for each content is added such that the column has value one in the time slots that are fixed to one so far, and zero in the other time slots. This is accomplished by Line~\ref{addingcol2}. Note that the fixed variables remain in effect in all subsequent iterations of RCGA.
\begin{algorithm}\label{alg_ERA}
\caption{Tailored Rounding Algorithm (TRA)}
\begin{algorithmic}[1]
\algsetup{linenosize=\tiny}
\small
\REQUIRE $\bm{w}^*$, $x_{tf}^{(k)}, t \in \mathcal{T}, f \in \mathcal{F}, k \in \mathcal{K}_f$
\STATE Compute $\bm{z}=\{z_{tf}, t \in \mathcal{T}, f \in \mathcal{F}\}$, where $z_{tf}=\sum_{k\in \mathcal{K}^\prime_f}x^{(k)}_{tf}w^*_{fk}$
\STATE   $x_{tf}\leftleftarrows1$ in SP$_f$ if $z_{tf}=1$, $t \in \mathcal{T}, f \in \mathcal{F}$ \label{fixxto1}
\STATE   $y_{fk}\leftleftarrows0$ in RMP if $x^{(k)}_{tf}=0$, $k \in \mathcal{K}^\prime_f, t \in \mathcal{T}, f \in \mathcal{F}$\label{fixyto0}
\STATE $\munderbar{z}\leftarrow\underset{t\in\mathcal{T}, f\in \mathcal{F}~~~~~~~~~~~~~~~~~~~~~~~~~~~} {\min\{z_{tf}| z_{tf}>0~\text{and}~z_{tf}<1\}}$\label{minz}
\STATE $(\munderbar{t},\munderbar{f})\leftarrow\underset{t\in\mathcal{T}, f\in \mathcal{F}~~~~~~~~~~~~~~~~~~~~~~~~~~~~~~~} {\arg\min\{z_{tf}| z_{tf}>0~\text{and}~z_{tf}<1\}}$\label{minzloc}

\STATE $\bar{z}\leftarrow\underset{t\in\mathcal{T}, f\in \mathcal{F}~~~~~~~~~~~~~~~~~~~~~~~~~~~~~~~~} {\min\{1-z_{tf}| z_{tf}>0~\text{and}~z_{tf}<1\}}$\label{maxz}
\STATE $(\bar{t},\bar{f})\leftarrow\underset{t\in\mathcal{T}, f\in \mathcal{F}~~~~~~~~~~~~~~~~~~~~~~~~~~~~~~~~~~~~} {\arg\min\{1-z_{tf}| z_{tf}>0~\text{and}~z_{tf}<1\}}$\label{maxzloc}

\IF {$(\munderbar{z}<\bar{z})$} \label{nearestcheck}
\STATE  $x_{\munderbar{t}\munderbar{f}}\leftleftarrows0$ in SP$_{\munderbar{f}}$\label{fixxunderbarto0}
\STATE   $y_{\munderbar{f}k}\leftleftarrows0$ if $x^{(k)}_{\munderbar{t}\munderbar{f}}=1$, $k \in \mathcal{K}^\prime_{\munderbar{f}}$\label{fixyunderbarto0}

\ELSIF{$(l_{\bar{f}}\le S^\prime)$}
\STATE  $x_{\bar{t}\bar{f}}\leftleftarrows1$ in SP$_{\bar{f}}$\label{fixxbarto1}
\STATE   $y_{\bar{f}k}\leftleftarrows0$ if $x^{(k)}_{\bar{t}\bar{f}}=0$, $k \in \mathcal{K}^\prime_{\bar{f}}$\label{fixybarto0}
\ELSE
\STATE  $x_{\bar{t}\bar{f}}\leftleftarrows0$ in SP$_{\bar{f}}$\label{fixxbarto0}
\STATE   $y_{\bar{f}k}\leftleftarrows0$ if $x^{(k)}_{\bar{t}\bar{f}}=1$, $k \in \mathcal{K}^\prime_{\bar{f}}$\label{fixybarto01}
\ENDIF
\FOR{$t=1$ ~to~ $T$}
\STATE $\mathcal{F}^\prime \leftarrow \{f \in \mathcal{F}| x_{tf} \text{ is fixed to one}\}$
\STATE $S^\prime\leftarrow S-\sum_{f \in \mathcal{F}^\prime}l_f$
\FOR{$f \in \mathcal{F}\backslash \mathcal{F}^\prime$}
\IF {$l_f> S^\prime$}\label{fixto0bysize1}
\STATE  $x_{tf}\leftleftarrows0$ in SP$_{f}$
\STATE   $y_{fk}\leftleftarrows0$ in RMP if $x^{(k)}_{tf}=1$, $k \in \mathcal{K}^\prime_{f}$ \label{fixto0bysize2}
\ENDIF
\ENDFOR
\ENDFOR
\FOR{$f=1$ ~to~ $F$}
\STATE $\mathcal{K}^\prime_{f} \leftarrow
\mathcal{K}^\prime_{f}\cup \{[x_{1f},\dots,x_{Tf}]^T\}$ where $x_{tf}=1$, $t \in \mathcal{T}$, if $x_{tf}$ is previously fixed to one and $x_{tf}=0$ otherwise \label{addingcol2}
\ENDFOR
\end{algorithmic}
\end{algorithm}

\subsection{Framework of RCGA}\label{sec:CGAandERA}
 Note that as none of the variables in the SPs or RMP is fixed  when CGA is applied for the first time (i.e., in the first iteration of Algorithm \ref{alg_CGAandERA}), the cost from CGA provides a lower bound to the global optimum of SCCD. This lower bound can be used to measure the effectiveness of the final solution from Algorithm \ref{alg_CGAandERA} or the solution obtained from any other suboptimal algorithm. The RCGA framework is shown in Algorithm~$\ref{alg_CGAandERA}$. The maximum number of iterations required to obtain a feasible solution is bounded by $F\times T$. Because, each time TRA is used, at least the caching decision of one content in one time slot is made, and as there are $F$ contents and $T$ time slots, Algorithm~$\ref{alg_CGAandERA}$ terminates in at most $F\times T$ iterations.

\begin{algorithm}\label{alg_CGAandERA}
\caption{Framework of RCGA}
\begin{algorithmic}[1]
\algsetup{linenosize=\tiny}
\small
\STATE STOP $\leftarrow 0$
\WHILE{(STOP$=0$)}
\STATE Apply CGA with fixed variable values so far and obtain $\bm{w}^*$
\IF {($\bm{w}^*$ is an integer solution)}
\STATE STOP $\leftarrow 1$
\ELSE
\STATE Apply TRA to $\bm{w}^*$
\ENDIF
\ENDWHILE
\end{algorithmic}
\end{algorithm}

\section{Greedy Algorithms}\label{sec:greed_algs}
In this section, we consider cheap algorithms. We propose two greedy algorithms that deal with one time slot at a time. These algorithms are developed based on two conventional caching algorithms in the literature, i.e., popularity-based caching (PBC) \cite{Ahlehagh2014Video} and random-based caching (RBC) \cite{Balaszczyszy2015Optimal}. In PBC, a content is chosen as a candidate to be stored in the cache based on how frequently it is requested. In RBC, the candidate content will be chosen randomly and proportionally to its popularity. That is, the higher a requested content is, the more likely this content will be selected as a candidate content.
Popularity of content $f$ in time slot $t$ is modeled by the  total number of the requests that must to be satisfied in this time slot, namely, all requests with deadline $t$. Denote by $\mathcal{P}_{tf}$ the set of these requests for content $f$ in time slot $t$. Denote by $P_{tf}$ the cardinality of set $\mathcal{P}_{tf}$. $\mathcal{P}_{tf}$ can be computed as:
\begin{equation}
\begin{aligned}
\mathcal{P}_{tf}=\{(u,r):u \in \mathcal{U} , r \in\mathcal{R}_u , h(u,r)=f, d(u,r)=t\}
\end{aligned}
\end{equation}

The flow of the two algorithms is similar and a general description is as follows. The time slots will be considered one by one starting from the first time slot. The cache is initialized with size of $S$ units of spare capacity. For each time slot under consideration, the algorithms treat contents one by one based on popularity in PBC and randomness in RBC. Once a content is selected as a candidate to be stored in the cache, the algorithms use an updating strategy based on the one in \cite{7524381} to decide whether to store the content in this time slot. The updating strategy is as follows. For candidate content $f$, one of the following scenarios may arise:

\begin{enumerate}
	\item If there is no enough spare space in the cache to store content $f$, the algorithms set $x_{tf}=0$.
    \item If the cache has enough spare space and the content was stored in the previous time slot, the decision is to keep the content, i.e., $x_{tf}=1$.
    \item If there is enough spare space but the content needs to be downloaded from the server, then the algorithms  store the content if it is at least as popular as some of the stored contents in the previous time slot. Specifically, content $f$ should be at least popular as the least popular contents with total size similar to $l_f$. This comparison is due to the fact that storing the candidate content leads to deleting the contents that were in the cache in the previous time slot. Thus, it is beneficial to put this content in the cache only if it is at least as popular as them.
\end{enumerate}
The flow of the two algorithms is shown in Algorithm~\ref{alg:greedy}.

\begin{algorithm}\label{alg:greedy}
\caption{The flow of PBC and RBC}
\begin{algorithmic}[1]
\algsetup{linenosize=\tiny}
\small
\REQUIRE $S$, $l_f$, $c_b$, and $c_s$
\ENSURE $\bm{x}$
\STATE $x_{f0} \leftarrow  0, \forall f \in \mathcal{F}$

\FOR{$t=1$ ~to~ $T$}\label{greedy:timeslot}
\STATE $S^\prime \leftarrow S$\label{greedy:cachesize}
\STATE Calculate $\mathcal{P}_{tf}, f \in \mathcal{F}$\label{greedy:pop}
    \STATE PBC: sort contents based on their popularity and put them in the sorted order in set $\mathcal{F}$\label{greedy:sort}
  \STATE  RBC: select contents randomly proportionally to their popularity and put following resulting order in set $\mathcal{F}$\label{greedy:sort}

\FOR{$f=1$ to $F$}
        \IF {$l_f> S^\prime$}
            \STATE $x_{tf} \leftarrow 0$\label{greedy:cond1}
        \ELSIF{($l_f\le S^\prime$ and $x_{(t-1)f}=1$)}
            \STATE $x_{tf} \leftarrow 1$\label{greedy:cond2.1}
            \STATE $S^\prime \leftarrow S^\prime-l_f$\label{greedy:cond2.2}
        \ELSIF{($l_f\le S$ and $x_{(t-1)f}=0$)}
                \STATE $\Psi \leftarrow\{i \in \{f+1,\dots,F\}|x_{(t-1)i}=1\}$\label{greedy:cond3.1}
\STATE $E^{del}\leftarrow0$\label{greedy:cond3.2}
\STATE $l^{del}\leftarrow0$
        \WHILE{($l^{del}\le l_f$ and $\abs{\Psi}>0$)}
        \STATE $E_{t}^{del}\leftarrow E_{t}^{del}+ \underset{f \in \Psi}\min\{\mathcal{P}_{tf}\}$
\STATE $f^\prime \leftarrow \underset{f \in \Psi}\argmin\{\mathcal{P}_{tf}\}$
\STATE $l^{del}\leftarrow l^{del}+l_{f^\prime}$
\STATE $\Psi \leftarrow \Psi  \setminus \{f^\prime\}$
        \ENDWHILE\label{greedy:cond3.3}

        \IF {$P_{tf}\ge E^{del} $}\label{greedy:cond3.4}
            \STATE $x_{tf} \leftarrow 1$
             \STATE $S^\prime \leftarrow S^\prime-l_f$
            \ELSE
            \STATE $x_{tf} \leftarrow 0$
            \ENDIF\label{greedy:cond3.5}
        \ENDIF

    \ENDFOR
\ENDFOR
\RETURN $\bm{x}$
\end{algorithmic}
\end{algorithm}

\section{Performance Evaluation}\label{sec:performance}
In this section, we conduct simulations to evaluate the performance of RCGA, PBC, and RBC by comparing them to the lower bound of global optimum; the lower bound is hereafter referred to as LB. As explained in Section \ref{sec:CGAandERA}, the LB is provided by the solution of the first iteration of Algorithm \ref{alg_CGAandERA}. In general, deviations of RCGA, PBC, and RBC from global optimum are hard to obtain, because it is difficult to calculate the global optimum of SCCD as it is an NP-hard problem. Therefore, we use the LB to measure the effectiveness of the algorithms because the deviation
to the global optimum cannot exceed the deviation to the LB. 
Hereafter, we refer to the relative deviations of RCGA, PBC, and RBC from LB as the (worst-case) optimality gaps.

\subsection{Simulation Setup}
For the simulation setup, we set $T=24$ where each time slot has a length of one hour \cite{Shukla2017Proactive,GhafourAhani2018}. Similar to the works in \cite{8629363,6883600}, we set $U=600$ and $F=200$ where the sizes of contents are uniformly generated within interval $[1,10]$. The capacity of the cache is set as $S=\rho \sum_{f \in \mathcal{F}}l_f$. Here, $\rho \in [0,1]$ is a parameter that shows the size of cache in relation to the total size of all contents. The number of requests for each user is uniformly distributed in interval $[1,10]$. $o(u,r)$, $u \in \mathcal{U}$ and $r\in \mathcal{R}_u$,  are randomly selected between time slots $1$ and $T$. The deadlines of content requests are uniformly selected in interval $[o(u,r),\alpha(T-o(u,r))]$ in which $\alpha$ indicates the tightness of deadlines. We will show the impact of $\alpha$ on the system cost.

 Same as many works (e.g.,\cite{KarthikeyanShanmugam2013}) in the literature, the content popularity distribution is modeled by a ZipF distribution, i.e., the probability that a user requests the $f$-th content is $\frac{f^{-\gamma}}{\sum_{i \in \mathcal{F}}i^{-\gamma}}$. Here $\gamma$ is the shape parameter of the ZipF distribution and is set to $\gamma=0.56$ \cite{KarthikeyanShanmugam2013}. The requests for contents
 are generated with varying content popularity over time.
 We will vary the parameters $\alpha$, $T$, $U$, $F$, $\rho$, and $\gamma$ in the simulations to show their impact on the system cost. Table~\ref{table:PE} summarizes the definitions of parameters for reference.

\begin{table}[ht!]
\centering
\caption{Definition of Parameters.}
\begin{tabular}{cl}
\hline
Symbol&Definition\\
\hline
$T$                     & number of time slots\\
$U$                     & number of users\\
$F$                     & number of contents\\
$S$                     & cache capacity\\
$\alpha$                & tightness of deadlines\\
$\rho$                  &cache capacity in relation to the total size \\
&of contents\\
$\gamma$               & shape parameter of ZipF distribution\\
$c_s$                     &downloading cost form server\\
$c_b$                     &downloading cost from base station
\\ \hline
\end{tabular}

\label{table:PE}
\end{table}

\subsection{ Performance Comparison}
The performance results of algorithms are reported in Figs.~$\ref{impact_deadline}\text{-}\ref{impact_pop}$. The lines in black, green, blue, and red represent the costs originating from the LB, RCGA, PBC, and RBC, respectively. The curves of RCGA and the LB are virtually overlapping in all figures, and the optimality gap of RCGA is consistently at most 1.6\%, thus the RCGA performance is impressive when it comes to solution quality.

Fig.~\ref{impact_deadline} shows the impact of tightness of deadlines on the cost. When $\alpha$ increases from $0$ to $1$, the costs obtained from RCGA, PBC, and RBC decrease by $31.9\%$, $35.9\%$, and $33.5\%$, respectively. The reason is that with less stringent deadline, the system has more opportunities to satisfy the requests via caching. The optimality gap of RCGA increases slightly from $0.6\%$ to $1.6\%$, while the corresponding values for PBC and RBC decrease from $26.1\%$ and $27.2\%$ to $19.3\%$ and $24.6\%$, respectively.
\begin{figure}[t]
\centering
\includegraphics[scale=0.48]{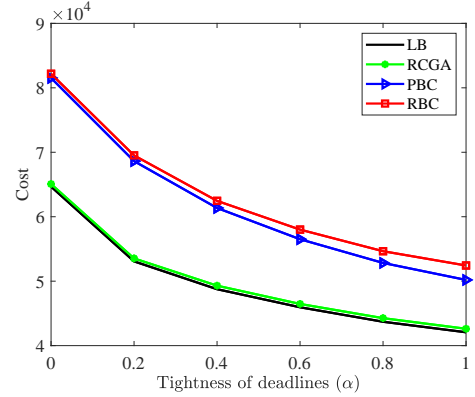}
\begin{center}
\caption{Impact of $\alpha$ on cost when $T=24, U=600$, $F=200$, $\rho=0.5$, $\gamma=0.56$, $c_s=10$,\text{~and~}$c_b=1$.}
\label{impact_deadline}
\end{center}
\end{figure}

Fig.~\ref{impact_T} shows the impact of number of time slots on the cost. The costs decrease with respect to the number of time slots. There are two reasons for this: With larger $T$ a) there are more opportunities to update contents of the cache, and b) more requests can be satisfied via the cache during the time period. The optimality gap of RCGA  stays always less than 1\%. However, the gap for PBC is $9.6\%$ for $T=6$ and increases to $20.1\%$ for $T=36$. The reason is that with larger $T$, the problem becomes more difficult which results in a higher optimality gap. The gap from RBC stays around 20.8\% for all values of $T$.

\begin{figure}[t]
\centering
\includegraphics[scale=0.48]{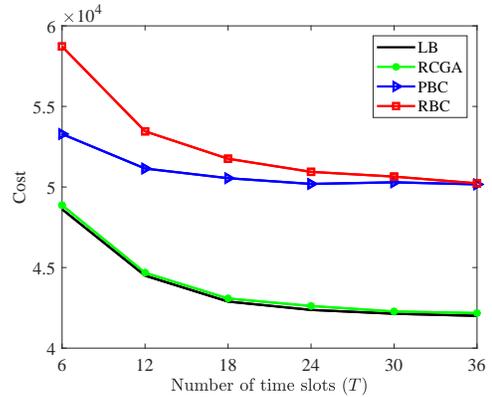}
\begin{center}
\caption{Impact of $T$ on cost when $U=600$, $F=200$, $\rho=0.5$, $\gamma=0.56$, $\alpha=1, $$c_s=10$,\text{~and~}$c_b=1$.}
\label{impact_T}
\end{center}
\end{figure}

Figs~\ref{impact_U} and \ref{impact_F} show the impact of $U$ and $F$ on the cost respectively. As can be seen, the cost increases with respect to $U$ and $F$. Obviously, this is because with larger $U$, the total number of requests increases accordingly which leads to a higher cost. Also, when $F$ increases, the diversity of requested contents increases, and as the cache capacity is limited, more requests need to be downloaded from the server which leads to a higher cost.
In general, the optimality gaps of RCGA, PBC, and RBC are approximately $1\%$, $18.5\%$, and $19.5\%$, for all values of $U$, respectively. The gaps of all algorithms slightly increase with respect to $U$ and this is more apparent for RBC. We can say that even if the size of problem increases with $U$ (i.e., more difficult), the solution quality of algorithms slightly decreases.

Increasing $F$ from $100$ to $300$, the optimality gap of RCGA decreases from $1.6\%$ to $0.2\%$, while the optimality gaps of PBC and RBC increase from $15.4\%$ and $18.6\%$ to $19.8\%$ and $21.1\%$, respectively. This shows that RCGA can effectively utilize the cache capacity, while PBC and RBC are not able to achieve this. In fact, with larger $F$, the diversity of requests increases and the problem becomes more challenging.
\begin{figure}[t]
\centering
\includegraphics[scale=0.48]{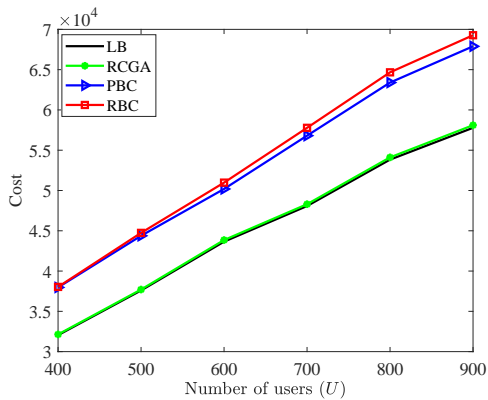}
\begin{center}
\caption{Impact of $U$ on cost when $T=24$, $F=200$, $\rho=0.5$, $\gamma=0.56$, $\alpha=1$, $c_s=10$,\text{~and~}$c_b=1$.}
\label{impact_U}
\end{center}
\end{figure}

\begin{figure}[t]
\centering
\includegraphics[scale=0.48]{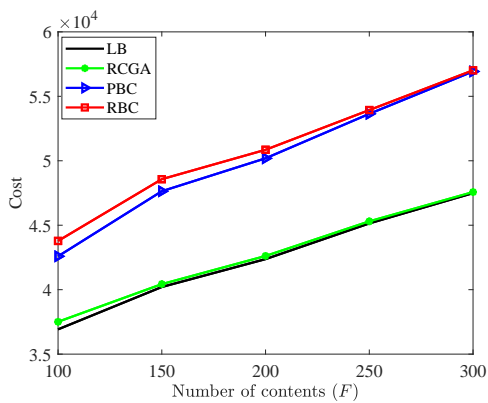}
\begin{center}
\caption{Impact of $F$ on cost when $T=24, U=600$, $\rho=0.5$, $\gamma=0.56$, $\alpha=1$, $c_s=10$,\text{~and~}$c_b=1$.}
\label{impact_F}
\end{center}
\end{figure}

Fig.~\ref{impact_rho} shows the effect of cache size in relation to the total size of contents. Overall, it can be observed that when $\rho$ grows from $0.1$ to $0.9$, the cost and optimality gaps obtained from RCGA, PBC, and RBC all decrease. This is due to the fact that a cache with more space can store more contents. RCGA outperforms both PBC and RBC and has nearly optimal solutions. The optimality gaps of RCGA, PBC, and RBC for $\rho=0.1$ are $1.4\%$, $21.1\%$, and $35.5\%$ respectively and they decrease to $0.1\%$, $4.5\%$, and $4.5\%$ when $\gamma$ increases to $0.9$. The reason is that when $\gamma=0.1$, the capacity is extremely limited, and it is crucial to utilize the capacity efficiently. RCGA is able to achieve this compared to PBC and RBC. When the caching space increases, the costs and optimality gaps start to decrease. When the caching space becomes excessively large such that most of the requested contents can be stored in the cache, optimizing the caching space becomes rather a trivial task and all algorithms have similar performance.

\begin{figure}[t]
\centering
\includegraphics[scale=0.48]{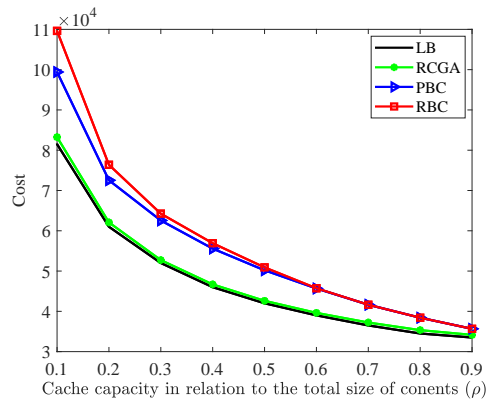}
\begin{center}
\caption{Impact of $\rho$ on cost when $T=24, U=600$, $F=200$, $\gamma=0.56$, $\alpha=1$, $c_s=10$,\text{~and~}$c_b=1$.}
\label{impact_rho}
\end{center}
\end{figure}

Finally, Fig.~\ref{impact_pop} shows the impact of popularity of contents on the cost. As can be seen the costs and optimality gaps decrease with respect to $\gamma$. Note that when $\gamma$ increases, the popularities of contents become more distinct and thus it is easier for the algorithms to determine which contents should be stored in the cache in order to achieve low cost.

\begin{figure}[t]
\centering
\includegraphics[scale=0.48]{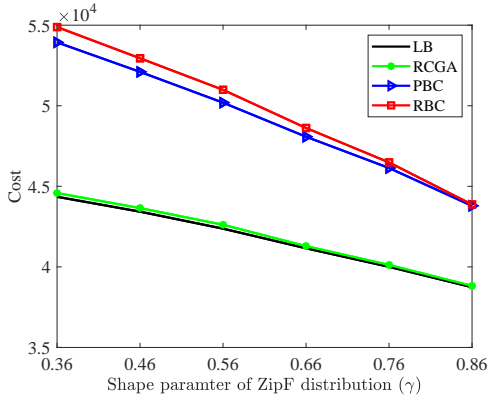}
\begin{center}
\caption{Impact of $\gamma$ on cost when $T=24, U=600$, $F=200$, $\rho=0.5$, $\alpha=1$, $c_s=10$,\text{~and~}$c_b=1$.}
\label{impact_pop}
\end{center}
\end{figure}

\section{Conclusions}\label{sec:conclo}
This paper has investigated a content caching problem where the joint impact of content downloading cost and deadline constraints are accounted for. First, the problem is formulated as an integer linear program (ILP). Even though the ILP can provide optimal solutions, it needs significant computational time for large-scale problem instances. Thus, three algorithms are developed for problem solving. The first one is a solution approach based on a repeated column generation algorithm (RCGA). The second and third algorithms are developed from popularity-based (PBC) and random-based caching (RBC) from the literature. PBC and RBC are simple and fast and thus they are suitable for very large-size problem instances. Simulation results have demonstrated that RCGA outperforms PBC and RBC algorithms and provides nearly optimal solutions within approximately $1.6\%$ gap of global optimum. In addition, simulation results show that one-third of the system cost can be cut off when content requests have longer deadlines. PBC and RBC are suitable for the scenarios when the cache capacity is fairly large or the popular contents are apparent, because for such scenarios they can provide solutions with qualities nearly the same as RCGA.





\bibliographystyle{IEEEtran}
\bibliography{IEEEabrv,ForIEEEBib}
\end{document}